\documentclass[letterpaper]{article} 
\usepackage[submission]{aaai2026}  
\usepackage{times}  
\usepackage{helvet}  
\usepackage{courier}  
\usepackage[hyphens]{url}  
\usepackage{graphicx} 
\usepackage{natbib}  
\usepackage{caption} 
\usepackage{algorithm}
\usepackage{algorithmic}

\usepackage{newfloat}
\usepackage{listings}
\DeclareCaptionStyle{ruled}{labelfont=normalfont,labelsep=colon,strut=off} 
\floatstyle{ruled}
\newfloat{listing}{tb}{lst}{}
\floatname{listing}{Listing}
\newif\iflong
\newif\ifshort
\newif\ifjournal

\longtrue

\iflong
\else
\shorttrue
\fi

\usepackage{complexity}
\usepackage{xcolor}
\usepackage{amsfonts}
\usepackage{amsthm}
\usepackage{amsmath}
\usepackage{mathtools}

\newcommand{\mpc}{\textsc{Maximum Price Coverage}}
\newcommand{\mpcShort}{\textsc{MPC}}
\newcommand{\mrbc}{\textsc{Red-Blue Reinforcement}}
\newcommand{\mrbcShort}{\textsc{R-BR}}

\newtheorem{theorem}{Theorem}

\newtheorem{property}{Property}

\newcommand{\qedclaim}{\hfill $\diamond$ \medskip}
\newenvironment{sketch}{\noindent{\it Sketch of proof.}}{\qedclaim}

\newcommand{\defproblem}[3]{
    \vspace{3mm}
    \noindent\fbox{
        \begin{minipage}{0.45\textwidth}
            #1\newline
            \textbf{Input:} #2\\
            \textbf{Task:} #3
        \end{minipage}
    }
    \vspace{3mm}
}

\title{Exact Algorithms for Resource Reallocation Under Budgetary Constraints}
\author {
    Arun Kumar Das\textsuperscript{\rm 1},
    Sandip Das\textsuperscript{\rm 2},
    Sweta Das\textsuperscript{\rm 2},
    Foivos Fioravantes\textsuperscript{\rm 3},
    Nikolaos Melissinos\textsuperscript{\rm 4},
}
\affiliations {
    \textsuperscript{\rm 1}University of Hyderabad\\
    \textsuperscript{\rm 2}Indian Statistical Institute\\
    \textsuperscript{\rm 3}Czech Technical University\\
    \textsuperscript{\rm 4}Charles University\\
}

\usepackage{bibentry}

\begin{document}

\maketitle
\begin{abstract}

Efficient resource (re-)allocation is a critical challenge in optimizing productivity and sustainability within multi-party supply networks. In this work, we introduce the \textsc{Red-Blue Reinforcement} (R-BR) problem, where a service provider under budgetary constraints must minimize client reallocations to reduce the required number of servers they should maintain by a specified amount. We conduct a systematic algorithmic study, providing three exact algorithms that scale well as the input grows (FPT), which could prove useful in practice. Our algorithms are efficient for topologies that model rural road networks (bounded distance to cluster), modern transportation systems (bounded modular-width), or have bounded clique-width, a parameter that is of great theoretical importance. 
\end{abstract}

\section{Introduction}
Efficient \emph{resource reallocation} has risen as a fundamental challenge towards enhancing productivity and sustainability~\cite{gupta2023resource,syrquin1984resource}. In a multi-party supply network, digital or otherwise, a vast number of clients and servers interact, where servers fulfill client demands by providing essential resources (e.g., data, goods, services, computational power, etc.). 

In this paper, we focus on the scenario where a service provider is forced to work under restrictions that do not allow for the serving of all clients. In such instances, we consider the optimization problem faced by the provider, who has to minimize the number of clients who get \textit{reallocated} in order 
to efficiently serve the remaining clients. A client is reallocated if a new connection is built between that client and an existing server. Clearly, building such connections can be costly, hence the need for minimizing these reallocations. 
Specifically, we seek the minimum number of clients whose reallocation leads to a reduction in the required number of servers supported by the considered provider by a specified amount. We model this scenario as an optimization problem defined on a network. On a high level, we are given an abstract network whose nodes are either red, or blue, or of both colors, signifying that the corresponding real nodes are servers, clients, or fulfill both roles. The links between these colored nodes signify the service relationships between servers and clients. 
Given a budget of $\gamma$ servers, our task is to find the minimum number of clients whose reallocation to a different provider allows for $\gamma$ servers to efficiently serve all the remaining clients. That is, we are looking for the minimum number of blue nodes, whose removal results in a network where $\gamma$ red nodes are linked to all the remaining blue nodes. Note that removing these nodes is effectively the same as adding new connections. We postpone the formal definition to the preliminaries section. 

Reallocating clients as a means of reducing the number of servers is not a novel idea. Examples of this strategy can be seen in, e.g., the \emph{team member reallocation}~\cite{agmon2005team} and \emph{facility reallocation}~\cite{fotakis2019reallocating} problems. The former focuses on leader minimization in coordinated task execution. A group of robots is assigned to perform a task, with a few of them designated as leaders, while the remaining robots function as team members under different leaders. The objective is to determine the minimum number of team members that must be reallocated to different leaders so that the total number of leaders required to successfully complete the operation decreases by at least a specified amount. The latter considers the case where a set of facilities is responsible for serving a group of clients, and a reallocation strategy must be devised respecting some specific cost constraints. We stress here that in this version, no server is allowed to simultaneously act as a client. 

Apart from its practical applications, the resource reallocation problem that we formulate is also interesting from a theoretical point of view as it generalizes the problem of computing the so-called \emph{reinforcement number}~\cite{blair2008domination} of a network, which is itself a generalization of the well-known \textit{domination number}. Simply put, given a network that can be dominated by at least $\gamma$ nodes (all the nodes of the network are linked to at least one of these $\gamma$ special nodes), 
the reinforcement number of that network is the minimum number of nodes whose removal results in a network that can be dominated by strictly less than $\gamma$ nodes. In other words, our problem coincides with computing the reinforcement number if all nodes in the given network are both red and blue; thus, we name the problem that we introduce in this work as the \mrbc{} (\mrbcShort{} for short) problem. 

\paragraph{Our contribution}
In this work, we introduce the \mrbcShort{} problem and initiate its formal study. Our work focuses on exact algorithms. That is, the aim is to construct algorithms that return an optimal solution with an accompanying theoretical guarantee of correctness. However, given the aforementioned link between our problem and the domination number, we can deduce that solving the \mrbcShort{} is computationally infeasible (\NP-hard). Thus, the central question becomes to identify for which topologies of networks this can be done efficiently. To achieve this, we employ the toolkit of \textit{parameterized complexity}, and construct \textit{Fixed Parameter Tractable} (\FPT) algorithms for specific structural parameters. 

Given a network with $n$ nodes, we present three results:
\begin{enumerate}
    \item An algorithm that runs in $3^{dc} n^{O(1)}$ time, where $dc$ is the \textit{distance to cluster} of the given network. 
    \item An algorithm that runs in $2^{mw} n^{O(1)}$ time, where $mw$ is the \textit{modular-width} of the given network. 
    \item An algorithm that runs in $4^{cw} n^{O(1)}$ time, where $cw$ is the \textit{clique-width} of the given network. 
\end{enumerate}

The choice of these parameters is not arbitrary. First, the parameter of distance to cluster captures how far the network is from being a collection of cliques. Such networks are important when considering the facility location problem in a large rural area. In this setting, the servers represent key facilities (e.g., hospitals, schools, etc.) which should be placed in strategic positions so that they can serve as many citizens as possible. It can be expected that such rural areas are comprised of many small and densely populated communities (the cliques) that are interconnected through a few highways. 

Then we study the parameter known as modular-width, which can be used to capture the hierarchical structure of modern transportation systems. Roughly speaking, a module in a network is a subset $M$ of nodes that share the same neighborhood outside $M$. In real-life terms, a module can be a ``neighborhood'', which is linked to the larger module of ``city'' through local public transportation. Then the module ``city'' is linked to the larger module of ``country'' through more large-distance means of transportation, and so on. Once more, we can expect that this hierarchical structure plays an important role in the placement of important facilities. See~\cite{HP10} for a survey on uses of modular decompositions. 

Finally, the clique-width is a fairly important parameter as it generalizes the widely used parameter of treewidth, which measures how ``tree-like'' is the input network is~\cite {CourcelleO00}. In contrast to the treewidth, the clique-width also captures a significantly larger family of networks, including many dense ones. In the current landscape of parameterized complexity, there are few parameters that are as general as the clique-width. Moreover, the running time of this third algorithm is optimal under the SETH~\cite{IP01}, painting a fairly comprehensive picture for the parameterized tractability of our problem. 

It is worth noting that structural graph parameters like the ones we consider here are increasingly being considered in multiple works on theoretical aspects of network analysis~\cite{blavzej2023parameterized,eiben2023parameterized}, path finding~\cite{fioravantes2024exact}, knowledge representation~\cite{gottlob2010bounded}, constraint satisfaction~\cite{ganian2022threshold}, planning~\cite{de2013parameterized}, etc.

\section{Preliminaries}

For $a\in \mathbb{N}$, let $[a] = \{1,\dots, a\}$.

The modeling of networks is done through the toolkit of graph theory. We follow standard graph-theoretic notation~\cite{D12}.
Let $G=(V,E)$ be a graph.
For a subset of vertices~\(U \subseteq V\) and a vertex $u \in V$ we denote by $N_{U}(u)$ the \textit{neighborhood} of $u$ in $U$, i.e., the set of vertices of~\(U\) that are adjacent to~\(u\).
By~\(N_U[u]\) we mean~\(N_U(u) \cup \{u\}\).
When clear from the context, the subscripts will be omitted.
Also, for any $S\subseteq V$, let $G[S]$ be the subgraph of $G$ that is induced by~$S$, i.e., the graph that remains after deleting the vertices of $V\setminus S$ from $G$ (along with their incident edges). Finally, for any set $S$ we denote the \textit{powerset} of $S$ by $\mathcal{P}(S)$.

Formally, the input of the \mrbcShort{} problem consists of a \textit{red-blue graph} $G=(V,E)$. That is, apart from $G$, we are also given two sets $R,B\subseteq V$ such the vertex set $V=R \cup B$. The sets $R$ and $B$ correspond to the red and blue vertices of $G$, respectively. We stress that $(R,B)$ is not necessarily a partition of $V$, i.e., it may be that $R\cap B\neq \emptyset$. To avoid overloading notations, we will implicitly assume that the sets $R$ and $B$ are well defined and given whenever we are considering a red-blue graph. 

The \textsc{Red-Blue Domination Number} of a red-blue $G$, denoted as $\gamma_{rb}(G)$, is the smallest integer such that there exists a set $D\subseteq R$ with $|D| = \gamma_{rb}$, and $D$ \textit{dominates} all vertices of $B$. That is, for every vertex $v\in B$ we have $N[v]\cap D \neq \emptyset$. 
The decision version of the problem we consider is as follows:

\defproblem{\mrbc{} (\mrbcShort)}{
    A red-blue graph $G$ and two positive integers $k$ and $\alpha$.
}{
    Decide if there exists a set $S\subseteq B$ such that $|S|\le k$ and $\gamma_{rb}(G[V\setminus S])\le  \gamma_{rb}(G)-\alpha$.
}

We may alternatively define an integer $\gamma:=\gamma_{rb}(G)-\alpha$ and answer the question of whether $\gamma_{rb}(G[V\setminus S])\le\gamma$. 


\subsection{Parameterized Complexity}
\emph{Parameterized complexity} is a domain of algorithm design which considers additional measures of complexity, known as \textit{parameters}. Roughly, it proposes a multidimensional analysis of the time complexity of an algorithm, each dimension corresponding to its own parameter.
Formally, a parameterized problem is a set of instances $(x,k) \in \Sigma^* \times \mathbb{N}$; $k$ is the \textit{parameter}.
Algorithmic efficiency in this paradigm is captured by the notion of \emph{Fixed-Parameter Tractability} (FPT). These are algorithms that solves the problem in $f(k)|x|^{O(1)}$ time for any arbitrary computable function $f\colon \mathbb{N}\to\mathbb{N}$. We say that a problem \textit{is in FPT} if it admits an FPT algorithm. We refer the interested reader to classical monographs~\cite{CyganFKLMPPS15,downey2012parameterized} for a more in-depth introduction to this topic.

\subsection{Structural Parameters}
There are three structural parameters that we consider in this work. Let $G=(V,E)$ be a graph. Unless otherwise specified, we set $n=|V|$ and $m=|E|$.

We begin with the \emph{distance to cluster} of $G$. A set $M\subseteq V$ is called a \emph{cluster deletion set} of $G$ if $G[V\setminus M]$ is a \textit{cluster}, i.e., a collection of disjoint cliques. The distance to cluster number of $G$ is the size of a minimum cluster deletion set of $G$. This is a parameter that can be computed in \FPT{} time, e.g., by~\cite{BCKP16}.

We then consider the \textit{modular-width} of $G$, initially introduced in~\cite{CMR00} and further studied in~\cite{GLO13}. A \textit{modular decomposition} of $G$ is a rooted tree $T_G$. 
Each node $b$ of $T_G$ represents a graph $G_b=(V_b,E_b)$, defined as follows:
\begin{itemize}
    \item If $b$ is a leaf node, then $G_b$ is an isolated vertex.
    \item Otherwise, $b$ has $c$ children and is related to a graph $H_b$ that has $c$ vertices $w_1,\dots,w_c$. Then $G_b$ is constructed as follows:  
    \begin{itemize}
        \item First, we consider the graphs $G_1,\ldots, G_c$, where $G_i$ is the graph represented by the $i^{th}$ child of $b$.
        \item We set $V(G_b) =\bigcup_{i \in [c]} V(G_i)$ and 
        \item $E(G_b) = \bigcup_{i\in [c]} E(G_i) \cup \{ uv \mid u \in V(G_i),v \in V(G_j) \textrm{ and } w_iw_j \in E(H_b)\}$. 
    \end{itemize} 
\end{itemize}
The \textit{width} of a modular decomposition $T_G$ is equal to the maximum number of children that any node of $T_G$ has.
A modular decomposition $T_G$ is a modular decomposition of $G$ if the root of $T_G$ represents $G$. 
A modular decomposition $T_G$ of minimum width can be computed in time $O(n+m)$~\cite{TCHP08} and has $O(n)$ nodes.

Finally, we have the \textit{clique-width} of $G$, whose definition passes through that of \textit{labeled graphs}. The graph $G$ is  \textit{labeled} if it is accompanied by a label function $\ell_G: V \rightarrow \mathbb{N}$. Given such a graph, we say that vertex $v$ has label $\ell_G(v)$. The clique-width of $G$ is the minimum number of labels needed to construct $G$ via labeled graphs over the following four operations: (1) \textit{introduce} a new \textit{vertex} with
label $\ell$, (2) \textit{rename} label $\ell$ to $\ell'$, (3) \textit{introduce all edges} between vertices labeled $\ell$ and $\ell'$, and (4) create the disjoint \textit{union} of two labeled graphs. 

\section{Efficient Algorithm I: Distance to Cluster}

In this section we prove that the \mrbc{} can be solved in \FPT{} time w.r.t. the distance to cluster of the input graph. 

Let $G=(V,E)$ be a red-blue graph, and $M$ be a cluster deletion set of $G$. 
On a high level, we first ``guess'' which of the vertices of $M$ will belong to the dominating set $D$. Then, we need to solve the problem of finding a set $U\subseteq V\setminus M$ that maximizes how many vertices are dominated if we include $U$ into $D$. To do so, we introduce the \mpc{} (\mpcShort{}), a variant of the classical \textsc{Maximum Coverage} problem. Roughly, we are given a set of elements belonging to different sets, each element and set having an accompanying \textit{price}. The goal is to identify the smallest number of sets that maximize the gathered price. We first give an efficient algorithm that solves \mpcShort{}. 
We then reduce the problem of computing a solution for \mrbcShort{} into computing a solution of \mpcShort{}.

\paragraph{An efficient algorithm for \mpcShort{}.}

We begin by formally defining the problem.

\defproblem{\mpc}{
    A set of elements $\mathcal{U}=\{e_1,\ldots,e_{\nu}\}$ called universe, a collection of sets $\mathcal{S}= \{S_1,\ldots, S_{\mu}\}$ of subsets of $\mathcal{U}$, a partition $\mathcal{X}=\{X_1,\ldots,X_{\ell}\}$ of $\mathcal{S}$, a price function $P:\mathcal{U}\cup \mathcal{X} \rightarrow \mathbb{N}$ and a positive integer $K$.
}{
    Find a set $C\subseteq \mathcal{S}$ of order at most $K$ of maximum price, i.e., that maximizes $pr(C)=\sum_{ e\in \bigcup S_i, S_i\in C } P(e) + \sum_{ X_i\cap C \neq \emptyset} P(X_i)$.
}

We develop an algorithm that solves \mpcShort{} in time $2^{\nu} (\nu+\mu)^{O(1)}$. To the best of our knowledge, this is the first time that this problem has been introduced. 
Our algorithm is heavily influenced by the classic dynamic programming algorithm for the \textsc{Set Cover} problem, which runs in exponential time in the number of elements~\cite{CyganFKLMPPS15}.
We stress, however, that some non-trivial modifications to this classic algorithm were needed to achieve our result. 

\begin{theorem}
    Let $(\mathcal{U}, \mathcal{S}, \mathcal{X}, P, K)$ be an instance of \mpcShort. There is an algorithm that computes a set $C\subseteq \mathcal{S}$ of maximum price in time $2^{|\mathcal{U}|}(|\mathcal{U}|+|\mathcal{S}|)^{O(1)}$.
\end{theorem}

\iflong
\begin{proof}
First, we reorder the sets of $\mathcal{S}$ into $S_1,\ldots,S_\mu$ such that if $S_i \in X_{i'} $ and $S_j \in X_{j'} $ for some $1\le i <j \le \mu$, then $i'\le j'$.
Also, hereafter we will denote by $X^i$ the set $X_{i'}\in \mathcal{X}$ that contains the set $S_i$. 

We define a table $F[U,q,i,x]$ where $U\subseteq \mathcal{U}$, $q \in [\mu]\cup \{0\}$, $i\in [\mu]$ and $x\in \{true,false\}$. 
In this table, for any $i\ge 1$ and $x\in \{true,false\}$, we want to keep the \textit{maximum profit} $ \sum_{ X_i\cap C \neq \emptyset} P(X_i)$ that we can achieve from any set $C$ such that:
\begin{itemize}
    \item $\bigcup_{S_i\in C} S_i \supseteq U$,
    \item $|C|=q$,
    \item $C\subseteq \{S_1,\ldots,S_i\}$ and 
    \item $ X^{i} \cap C \neq \emptyset$ if and only if $x=true$.
\end{itemize}
For $i=0$, we assume that we have no sets to use. In the case that the task indicated by the cell of the table is impossible, the table keeps the value $-\infty$. 

First, we compute the values of the table for $i=1$. Notice that if $q=0$, we can only cover $U=\emptyset$. Also, we cannot use any set from $X^{1}$. Thus $F[\emptyset, 0, 1,true]=0$ while $F[U, 0,1,x]=-\infty$ for any $(U,x)\neq (\emptyset,false)$. Next, we consider the case where $q=1$.
Notice that for any set $U\subseteq S_1$, there is only one way to cover $U$. Then the profit is exactly $p(X^1)$. Thus, we set $F[U,1,1,true]=p(X^1)$ for all $U \subseteq S_1$. Any other set $U'\subseteq \mathcal{U}$ with $U'\neq U$ cannot be covered with just one set; thus $F[U',1,1,true]= - \infty$. 
Finally, there are not enough sets in order to consider the values $q\ge 2$; thus we set $F[U,q,1,x]=- \infty$ for all $U\subseteq \mathcal{U}$, $2\le q \le \mu$ and $x\in \{true, false\}$. 

Next, assuming that we have correctly computed the values of $F$ for all $U\subseteq \mathcal{U}$, $1\le i < j \le \mu$, $q\in [\mu]$ and $x\in \{true,false\}$, we will show how to compute the values of $F$ for $j$. 
Again, we start by setting $F[\emptyset, 0,j,false]=0$ and $F[U, 0,j,x]=-\infty$ for any $(U,x)\neq (\emptyset,false)$. We continue by considering values of $q\ge 1$.

First, we show how to compute the values of $F$ for $x=false$. In this case, we cannot use the set $S_j$; therefore, the maximum profit must appear in the values of $F$ for $j-1$. In particular, we need to consider whether $X^{j-1}=X_j$ or not. If $X^{j-1}=X_j$, we set $F[U,q,j,false]= F[U,q,j-1,false]$ as we should not use any set from $X^{j-1}=X_j$. 
Otherwise, $S_j$ is the only set of $X^j$ that appears in $\{S_1\ldots,S_j\}$ (by the ordering of the $S_i$'s); therefore we set  $F[U,q,j,false]= \max \{ F[U,q,j-1,true],F[U,q,j-1,false] \}$ as it does not matter if we are using sets from $X^{j-1}$.

We now show how to compute the values of $F$ for $x=true$. Again, we need to consider whether $X^{j-1}=X_j$ or not. Assume first that $X^{j-1}=X_j$. 
We claim that $F[U,q,j,true]= \max \{F[U,q,j-1,true], F[U \setminus S^j,q-1,j-1,true], F[U \setminus S^j,q-1,j-1,false] +P(X^j) \}$. 
Indeed, there are only three cases that achieve the maximum profit:
\begin{enumerate}
    \item We exclude $S_j$. Thus $S_j$ should already have appeared in $F[U,q,j-1,true]$ ($x=true$ because we have assumed that $X^{j-1}=X_j$). 
    \item We use $S_j$ and another set from $S_{j'}\in X^j$ such that $j'<j$; this should include $q-1$ sets (excluding $S_j$) and must cover $U\setminus S_j$. However, since we have already included a set from $X^j$, we do not increase the profit by including $S_j$. Thus, this profit is exactly the one in $F[U\setminus S_j ,q-1,j-1,true]$. 
    \item We use $S_j$ but no other set from $X^j$; this should include $q-1$ sets (excluding $S_j$) and must cover $U\setminus S_j$. Since no other set from $X^j$ is included, we need to add $P(X^j)$. This is exactly $F[U\setminus S_j ,q-1,j-1,false]+ p(X^j)$. 
\end{enumerate}
Then we assume that $X^{j-1}\neq X_j$. Since we want to set $x=true$, it is mandatory to include $S_j$. Thus, we set $F[U,q,j,true]= \max \{F[U\setminus S_j, q-1,j-1,true] +P(X^j), F[U\setminus S_j, q-1,j-1,false] +P(X^j)\}$. 
Notice that we need to consider both $F[U\setminus S_j, q-1,j-1,true]$ and $F[U\setminus S_j, q-1,j-1,false]$ because $X^{j-1}\neq X_j$ and we do not know whether or not the maximum profit can be achieved with or without any sets from $X^{j-1}$ (as $X^{j-1}\neq X_j$). 

At this point, observe that including more sets in $C$ does not decrease the profit. Therefore, it suffices to compute the value of $\max_{U \subseteq \mathcal{U}, q \in [k]} \big\{ \sum_{e\in U} P(e) + \max \{ F[U,q,K,true], F[U,q,K,false]\} \big\}$. 
Finally, observe that throughout this process, we have also described how to construct the set $C$ for each cell of the table. 

As for the running time, in order to compute the value of $F$ of each cell, we just need to consider a constant number of options. Therefore, the algorithm runs in $2^{\nu}(\nu+\mu)^{O(1)}$ time. 
\end{proof}
\fi

\paragraph{Reducing \mrbcShort{} to \mpcShort{}.}
We are now ready to describe the reduction from \mrbcShort{} to \mpcShort{}. 
\begin{theorem}
    Let $\mathcal{I}=(G,k,\gamma)$ be an instance of \mrbcShort{}. We can decide if $\mathcal{I}$ is a yes-instance in time $3^{dc}n^{O(1)}$, where $n=|V(G)|$ and $dc$ is the distance to cluster number of $G$.
\end{theorem}

\begin{proof}
We first compute a minimum cluster deletion set $M$ of $G$. This can be done in $3^{dc}n^{O(1)}$, where $dc=|M|$. Recall that the graph $G$ is a red-blue graph, with $R,B\subseteq V$ being the respective red and blue vertices. Then, $\mathcal{I}$ is a yes-instance of \mrbcShort{} if we can construct a set $S\subseteq B$ with $|S|\leq k$ such that there exists a set $D\subseteq R$ that dominates all the remaining blue vertices of $G[V\setminus S]$, and $|D|\leq \gamma$. 

We begin by guessing which vertices of $M$ will be included in $D$ (there are $2^{dc}$ such guesses); let $D_M$ be this set. 
Assuming that we have made the correct guess, we create an equivalent instance of \mrbcShort{}
as follows: 
\begin{enumerate}
    \item We delete all the vertices in $B\setminus R$ that are dominated by $D_M$. 
    \item We delete all the vertices in $(R\setminus B)\cap M$. These vertices are redundant, as they do not belong to $B$ and do not need to be included in $D_M$.
    \item We remove from $R$ any vertex that belongs to $B\cap M$. Notice that we need to keep these vertices in $B$ as we do not know if they will be dominated by the final $D$ (or they will need to be included in $S$).
    \item We remove from $B$ any vertex that belongs to $R \setminus M$ and is dominated by $D_M$. Notice that we need to keep these vertices in $R$ as we may want to include them in the final $D$. 
    \item We remove $D_M$ from the graph (as these vertices will be in $D$ and they always dominate themselves).
\end{enumerate}
Let $G'=(V',E')$ be the new graph, $R'$, $B'$ be the new sets of red and blue vertices respectively, and $M'=M\setminus D_M$. Notice that $G'[V'\setminus M']$ is also a cluster. Moreover, if there exists a clique $Q$ in $G'[V'\setminus M']$ that does not include any red vertex, then we delete $V(Q)$ and reduce $k$ by $|Q|$.
Finally, we reduce $\gamma$ by $|D_M|$ as $D_M$ will be part of the final $D$. Let $\mathcal{I}'=(G',k',\gamma')$ be the new instance of \mrbcShort{}. Notice that $\mathcal{I'}$ is a yes-instance of \mrbcShort{} if and only if $\mathcal{I}$ is a yes-instance of \mrbcShort{}.

From the construction of $\mathcal{I}'$, we have, $M'\cap R = \emptyset$ and $V(Q)\cap R \neq \emptyset$, for any clique $Q$ in $G'[V'\setminus M']$. Thus, we need to decide how to optimally deal with these remaining, un-dominated, blue vertices. We do this by using our previous algorithm for \mpcShort{}.

We construct an instance $\mathcal{I}''=(\mathcal{U},\mathcal{S},\mathcal{X},P,K)$ of \mpcShort{} as follows. 
For each vertex $u \in M'$, we create one element $e_u \in \mathcal{U}$. For the family $\mathcal{S}$, for each red vertex $v \in R$, we create a set $S_v\in \mathcal{S}= \{e_u\mid u \in N_{G'}(v)\cap M' \}$, and add it to $\mathcal{S}$. Next, we define $\mathcal{X}$. For each clique $Q$ of $G'[V'\setminus M']$ we create a set $X_Q = \{S_v\mid v \in V(Q)\cap R\}$. Notice that each red vertex belongs to a unique clique of $G'[V'\setminus M']$. Thus, the defined $\mathcal{X}$ is indeed a partition of $\mathcal{S}$. The price for each partition is as follows: we set $P(X_Q)= |V(Q)\cap B|$ for each $X_Q\in \mathcal{X}$, and we set $P(e_u)=1$ for each element $e_u \in \mathcal{U}$. 
Finally, we set the maximum number of sets we can use to $K=k'$. 

To see the equivalence between $\mathcal{I}'$ and $\mathcal{I}''$, it suffices to notice the following. First, any feasible solution $C$ of the $\mathcal{I}''$ represents a set $D_C\subseteq R$. Also, the price $pr(C)$ is equal to the number of blue vertices dominated by $D_C$ in $G'$.
Finally, maximizing the number of vertices we cover in $C$ minimizes the number of vertices that cannot be dominated by $D_C$ (and thus will need to be deleted from $V'$).

It remains to prove the claimed running time. 
First, notice that both $|\mathcal{U}|$ and $|\mathcal{S}|$ are bounded by the number of vertices in $G'$. Therefore, the polynomial factors that appear in the running time of the \mpcShort{} algorithm can be replaced by a polynomial of $n$, i.e., $|V'|$.
Next, recall that we started by guessing $D_M$, the intersection of $D$ and $M$, and that there are $2^{dc}$ such intersections. Then, the instance of \mpcShort{} we created for this guess has a number of elements equal to $|M\setminus D|$  (as we create one element per vertex of $M$ that remains in the graph after the deletion of $D_M$). For any size of intersection $i\le dc$, we create $dc \choose i$ instances with $dc -i$ elements. Therefore, the running time is
$\sum_{i=0}^{dc} {dc \choose i} 2^{dc-i}n^{O(1)}  = 3^{dc} n^{O(1)}$.
\end{proof}

\section{Efficient Algorithm II: Modular-Width}
\begin{theorem}
    Let $\mathcal{I}=(G,k,\gamma)$ be an instance of \mrbcShort{}. We can decide if $\mathcal{I}$ is a yes-instance in time $2^{mw}(n)^{O(1)}$, where $n=|V(G)|$ and $mw$ is the modular-width of $G$.
\end{theorem}

\iflong
\begin{proof}

We will perform dynamic programming on the modular decomposition $T_G$ of $G$. Recall that $G=(V,E)$ is a red-blue graph, and that $R$ and $B$ are the red and blue vertices, respectively, of $G$. 

Let $b$ be a node of $T_G$ and $G_b=(V_b,E_b)$ be the graph represented by $b$. We will use $R_b$ and $B_b$ to denote the red and blue vertices, respectively, of $G_b$. We say that $G_b$ is \textit{dominable} by a set $D\subseteq R_b$ if $N[v]\cap D\neq \emptyset$, for all $v \in B_b$.  
For each node $b$, we define a function $f_b: [n]\cup \{0\} \rightarrow \mathcal{P} (B)$. This function takes as input an integer $\gamma$ and returns a set $S\subseteq V_b$ of minimum size such that $G_b\setminus S$ is dominable by a set $D$ of order at most $\gamma$.

Let $S=f_r(\gamma)$, where $r$ is the root of $T_G$. It is clear that $\mathcal{I}$ is a yes-instance of \mrbcShort{} if and only if $|S|\leq k$. We will now show how we can compute the function $f_b$ for any node $b$ of $T_G$. 

We first deal with the leaf nodes of $T_G$. Let $\ell$ be a leaf node and $V_\ell= \{v\}$. 
If $v\notin B$ then we set $f_\ell(\gamma)=\emptyset$ for all $\gamma\in [n] \cup \{0\} $. If $v \in B$ and $v\notin R$ then we need to delete $v$ as it cannot be dominated; thus we set $f_\ell(\gamma)=\{v\}$ for all $\gamma\in [n] \cup \{0\} $. Finally, if $B=R=\{v\}$, then we set $f_\ell(0)=\{v\}$ and $f_\ell(\gamma)=\emptyset$, for all $\gamma\in [n]$. 

Next, we consider the non-leaf nodes of $T_G$.
Recall that any non-leaf node $b$ is related to a graph $H_b$ with vertex set $V(H_b) = \{w_1,\ldots,w_c\}$ and we construct $G_b$ by using $H_b$ together with the $c$ graphs $G_{b_i}$, $i \in [c]$, each one represented by a child $b_i$ of $b$. Also, recall that $c \le mw$. 
We also assume that we have already computed the functions $f_{b_1},\ldots,f_{b_c}$. To simplify the notation, we denote by $f_i$ the function $f_{b_i}$, by $G_i$ the graph $G_{b_i}$, and by $V_i$ the set $V_{b_i}$, for all $i \in [c]$. We also use $B_i$ for $B\cap V_i$ and $R_i$ for $R\cap V_i$. 

We start by guessing which sets among the $V_i$'s will contain vertices of $D$. Let $X=\{i\in[c]|D\cap V_i\neq \emptyset\}$. Such a guess is \textit{valid} if $R_i \neq \emptyset$ for every $i\in X$. There are at most $2^c\le 2^{mw}$ valid guesses. For any valid guess $X$, we will compute a function $f^X_b$ which takes $\gamma$ as its input, and returns a set $S$ (of minimum order) such that:
\begin{itemize}
    \item there exists a set $D\subseteq R_b \setminus S$ of order $\gamma$,
    \item $D\cap V_i \neq \emptyset$, for all $i \in X$ and 
    \item $G_b\setminus S$ is dominable by $D$.
\end{itemize}
Finally, we set $f^{X}_b(\gamma) = B$ if $X$ is not valid or $0\leq \gamma<X$. 
Once we have computed $f^X_b$, for all $X\subseteq [c]$ and $\gamma \in [n]$, then we compute $f_b$ by setting $f_b(\gamma) =f^{X^*}_b(\gamma)$, where  $X^*$ is such that $|f^{X^*}_b(\gamma)| \le \min_{X\subseteq [c]}\{ | f^{X}_b(\gamma) | \}$.

Hereafter, we will show how to compute $f^X_b$ for a fixed valid guess $X\subseteq [c]$. 

First, consider any set $B_j$ such that $w_iw_j\in E(H_b)$ for some $i \in X$. Notice that any vertex $v \in R_i$ dominates all vertices in $B_j$. Therefore, we can ignore all such vertices as we can guarantee that they will be dominated by selecting at least one vertex in $R_i$, for each $i \in X$. 

Now, consider any set $B_j$ such that $w_iw_j\notin E(H)$ for all $i \in X$, $j\neq i$. These vertices can never be dominated. Therefore, we need to include them in the set $S$ that is the output of the function $f^X_b(\gamma)$, for all $\gamma\in [n] \cup \{0\}$. Let $S_X$ be the set of all such vertices. 

In order to compute the rest of the vertices we need to include in $S$, it suffices to decide the exact number of vertices we include in $D$ from each one of the $R_i$'s, for all $i \in X$. 
To do so, we first partition $X$ into two sets $X_1 = X \cap \{i \mid j \in X \textrm{ and } w_iw_j\in E(H_b)\}$ and $X_2 =X \setminus X_1$. 

For any $i \in X_1$, it is enough to include one vertex of $R_i$ in $D$. In fact, by the previous arguments, the vertices of $B_i$, for any $i \in X_1$, will be dominated regardless of which vertex of $R_i$ we choose to include in $D$. Additionally, all vertices of $R_i$ have the same neighborhood outside of $V_i$. Thus, it suffices, and is safe to select (arbitrarily) any one vertex of each $R_i$ and include it in $D$. 

It remains to decide how many vertices of each of $R_i$ will be included in $D$ for each $i \in X_2$. This choice is important in order to minimize the vertices of $B_i$ that will be included in $S$. 
In particular, we can select up to $\gamma - |X_1|$ vertices (as $D$ includes one vertex from each $R_i$, $i \in X_1$). 
Also, for any $i \in X_2$, a vertex $v \in B_i$ will either be dominated by a vertex in $R_i$ or deleted. 

Notice that $G[\bigcup_{i \in X_2}V_i]$ is the disjoin union of the graphs $G_i$, $i \in X_2$ (by the definition of $X_2$). Additionally, we have the restriction that the set $S$ that we are computing has at most $\gamma - |X_1|$ vertices and that it intersects $V_i$, for all $i \in X_2$. 

To compute the remaining vertices that will be included in $S$ we will compute a set of functions $f^X_{[p]}: [|R|]\setminus \{0,\ldots,p-1\} \rightarrow \mathcal{P}(B)$, for every $p \in [|X_2|]$. These functions take $\gamma$ as input, and return a set $S \subseteq \bigcup_{q \in [p]} B_{i_q}$ of minimum order such that $G[ \bigcup_{q \in [p]} V_{i_q} \setminus S ]$ is dominable by a set $D \subseteq \bigcup_{q \in [p]} R_{i_q}$ of order $\gamma$ that intersects all sets $ V_{i_q}$, $q \in [p]$.

Before we continue with the computation of these functions, we want to mention that this is the last step needed in order to compute $f^X_b$ (and therefore $f_b$). Indeed, if we compute $f^X_{[|X_2|]}$, then we have also computed $f^X_b(\gamma)$ since, by the previous arguments, $f^X_b(\gamma) = S_X \cup f^X_{[|X_2|]} (\gamma - |X_1|)$. 

The computation of the functions $f^X_{[p]}$, for every $p\in [|X_2|]$, is done recursively.
For $p=1$, we already have computed the values we need as $f^X_{[1]}(\gamma)= f_{i_1}(\gamma)$, for all $\gamma \ge 1$.

Assume that we have a computed the function $f^X_{[p]}$ for a $p \in [|X_2|-1]$. We will show how we can use $f^X_{[p]}$ together with $f_{i_{p+1}}$ in order to compute the $f^X_{[p+1]}$. 
For any $\gamma\ge [p+1]$, we have that
$
f^X_{[p+1]} (\gamma) = \min_{\gamma_1\ge p, \gamma_2\ge 1 \textrm{ and } \gamma_1+\gamma_2 = \gamma} \{f^X_{[p]}(\gamma_1) + f_{i_{p+1}}(\gamma_2) \}
$. 
This means that we can compute all values of $f^X_{[p+1]}$ in polynomial time and, therefore, we can also compute $f^X_{[|X_2|]}$ in polynomial time. This completes the computation of $f_b$ for all nodes $b$ of $T_G$. 

It remains to show that the algorithm runs in the claimed running time. 
For any leaf node $\ell$, we can compute $f_{\ell}$ in constant time.
Let $b$ be a non-leaf node of $T_G$. In order to compute $f_b$, we first guessed whether $D$ intersects with the vertex sets $V(G_1), \ldots, V(G_n)$, where $G_1, \ldots, G_n$ are the graphs represented by the children nodes of $b$. We have at most $2^n \le 2^{mw}$ such guesses. For each such guess, we computed (in polynomial time) a function $f^X_b$ (where $X \subseteq [c]$) that returns the set $S$. After computing $f^X_b$, for all $X\subseteq [c]$, it takes at most $2^{mw} n^{O(1)}$ time to compute $f_b$, since $f_b(\gamma)$ returns a smallest set between all the $f^X_b(\gamma)$'s, $X \subseteq [c]$. 
Finally, since we have $O(n)$ nodes in $T_G$, we need to compute at most $O(n)$ such functions, and this will take at most $2^{mw}n^{O(1)}$ time in total.
\end{proof}
\fi
\ifshort
\begin{sketch}
We will perform dynamic programming on the modular decomposition $T_G$ of $G$. Recall that $G=(V,E)$ is a red-blue graph, and that $R$ and $B$ are the red and blue vertices, respectively, of $G$. 

Let $b$ be a node of $T_G$ and $G_b=(V_b,E_b)$ be the graph represented by $b$. We will use $R_b$ and $B_b$ to denote the red and blue vertices, respectively, of $G_b$. We say that $G_b$ is \textit{dominable} by a set $D\subseteq R_b$ if $N[v]\cap D\neq \emptyset$, for all $v \in B_b$.  
For each node $b$, we define a function $f_b: [n]\cup \{0\} \rightarrow \mathcal{P} (B)$. This function takes as input an integer $\gamma$ and returns a set $S\subseteq V_b$ of minimum size such that $G_b\setminus S$ is dominable by a set $D$ of order at most $\gamma$.

Let $S=f_r(\gamma)$, where $r$ is the root of $T_G$. It is clear that $\mathcal{I}$ is a yes-instance of \mrbcShort{} if and only if $|S|\leq k$. We will now show how we can compute the function $f_b$ for any node $b$ of $T_G$.

We first deal with the leaf nodes of $T_G$. Let $\ell$ be a leaf node and $V_\ell= \{v\}$. 
If $v\notin B$ then we set $f_\ell(\gamma)=\emptyset$ for all $\gamma\in [n] \cup \{0\} $. If $v \in B$ and $v\notin R$ then we need to delete $v$ as it cannot be dominated; thus we set $f_\ell(\gamma)=\{v\}$ for all $\gamma\in [n] \cup \{0\} $. Finally, if $B=R=\{v\}$, then we set $f_\ell(0)=\{v\}$ and $f_\ell(\gamma)=\emptyset$, for all $\gamma\in [n]$. 

Next, we consider the non-leaf nodes of $T_G$.
Recall that any non-leaf node $b$ is related to a graph $H_b$ with vertex set $V(H_b) = \{w_1,\ldots,w_c\}$ and we construct $G_b$ by using $H_b$ together with the $c$ graphs $G_{b_i}$, $i \in [c]$, each one represented by a child $b_i$ of $b$. Also, recall that $c \le mw$. 
We also assume that we have already computed the functions $f_{b_1},\ldots,f_{b_c}$. To simplify the notation, we denote by $f_i$ the function $f_{b_i}$, by $G_i$ the graph $G_{b_i}$, and by $V_i$ the set $V_{b_i}$, for all $i \in [c]$. We also use $B_i$ for $B\cap V_i$ and $R_i$ for $R\cap V_i$. 

We start by guessing which sets among the $V_i$'s will contain vertices of $D$. Let $X=\{i\in[c]|D\cap V_i\neq \emptyset\}$. Such a guess is \textit{valid} if $R_i \neq \emptyset$ for every $i\in X$. There are at most $2^c\le 2^{mw}$ valid guesses. For any valid guess $X$, we will compute a function $f^X_b$ which takes $\gamma$ as its input, and returns a set $S$ (of minimum order) such that:
\begin{itemize}
    \item there exists a set $D\subseteq R_b \setminus S$ of order $\gamma$,
    \item $D\cap V_i \neq \emptyset$, for all $i \in X$ and 
    \item $G_b\setminus S$ is dominable by $D$.
\end{itemize}
Finally, we set $f^{X}_b(\gamma) = B$ if $X$ is not valid or $0\leq \gamma<X$. 
Once we have computed $f^X_b$, for all $X\subseteq [c]$ and $\gamma \in [n]$, then we compute $f_b$ by setting $f_b(\gamma) =f^{X^*}_b(\gamma)$, where  $X^*$ is such that $|f^{X^*}_b(\gamma)| \le \min_{X\subseteq [c]}\{ | f^{X}_b(\gamma) | \}$.

Hereafter, we will show how to compute $f^X_b$ for a fixed valid guess $X\subseteq [c]$. 

First, consider any set $B_j$ such that $w_iw_j\in E(H_b)$ for some $i \in X$. Notice that any vertex $v \in R_i$ dominates all vertices in $B_j$. Therefore, we can ignore all such vertices as we can guarantee that they will be dominated by selecting at least one vertex in $R_i$, for each $i \in X$. 

Now, consider any set $B_j$ such that $w_iw_j\notin E(H)$ for all $i \in X$, $j\neq i$. These vertices can never be dominated. Therefore, we need to include them in the set $S$ that is the output of the function $f^X_b(\gamma)$, for all $\gamma\in [n] \cup \{0\}$. Let $S_X$ be the set of all such vertices. 

In order to compute the rest of the vertices we need to include in $S$, it suffices to decide the exact number of vertices we include in $D$ from each one of the $R_i$'s, for all $i \in X$. 
To do so, we first partition $X$ into two sets $X_1 = X \cap \{i \mid j \in X \textrm{ and } w_iw_j\in E(H_b)\}$ and $X_2 =X \setminus X_1$. 

For any $i \in X_1$, it is enough to include one vertex of $R_i$ in $D$. In fact, by the previous arguments, the vertices of $B_i$, for any $i \in X_1$, will be dominated regardless of which vertex of $R_i$ we choose to include in $D$. Additionally, all vertices of $R_i$ have the same neighborhood outside of $V_i$. Thus, it suffices, and is safe to select (arbitrarily) any one vertex of each $R_i$ and include it in $D$. 

It remains to decide how many vertices of each of $R_i$ will be included in $D$ for each $i \in X_2$. This choice is important in order to minimize the vertices of $B_i$ that will be included in $S$. 
In particular, we can select up to $\gamma - |X_1|$ vertices (as $D$ includes one vertex from each $R_i$, $i \in X_1$). 
Also, for any $i \in X_2$, a vertex $v \in B_i$ will either be dominated by a vertex in $R_i$ or deleted. 

Notice that $G[\bigcup_{i \in X_2}V_i]$ is the disjoin union of the graphs $G_i$, $i \in X_2$ (by the definition of $X_2$). Additionally, we have the restriction that the set $S$ that we are computing has at most $\gamma - |X_1|$ vertices and that it intersects $V_i$, for all $i \in X_2$. 

To compute the remaining vertices that will be included in $S$ we will compute a set of functions $f^X_{[p]}: [|R|]\setminus \{0,\ldots,p-1\} \rightarrow \mathcal{P}(B)$, for every $p \in [|X_2|]$. These functions take $\gamma$ as input, and return a set $S \subseteq \bigcup_{q \in [p]} B_{i_q}$ of minimum order such that $G[ \bigcup_{q \in [p]} V_{i_q} \setminus S ]$ is dominable by a set $D \subseteq \bigcup_{q \in [p]} R_{i_q}$ of order $\gamma$ that intersects all sets $ V_{i_q}$, $q \in [p]$. 

The computation of the functions $f^X_{[p]}$, for every $p\in [|X_2|]$, is done by recursion on $p$, and follows similar arguments as the ones exploited so far.
Finally, it takes at most $2^{mw} n^{O(1)}$ time to compute $f_b$.
Since we have $O(n)$ nodes in $T_G$, we need to compute at most $O(n)$ such functions, and this will take at most $2^{mw}n^{O(1)}$ time in total.
\end{sketch}
\fi


\section{Efficient Algorithm III: Clique-Width}
\begin{theorem}\label{thm:cw}
Given an instance $\mathcal{I}=(G,k,\gamma)$ of \mrbcShort{} and a clique-width expression $T_G$ of $G$ using at most $cw$ labels, we can decide if $\mathcal{I}$ is a yes-instance in time $4^{cw}n^{O(1)}$. 
\end{theorem}

\iflong
\begin{proof}
We will perform dynamic programming over the clique-width expression $T_G$.
Let $G_t$ denote the graph constructed by the expression defined by the subtree rooted at node $t$, and let $V_t$ denote the vertex set of $G_t$. Our algorithm will compute a set $S\subseteq B$ of at most $k$ vertices such that there exists a set $D\subseteq R$ of at most $\gamma$ vertices that dominates all the blue vertices of $G\setminus S$, if such sets do indeed exist.  
For each $G_t$ associated with node $t$ of the given clique-width expression, the algorithm maintains information about:
\begin{itemize}
\item which red vertices are included in set $D$,
\item which blue vertices are included in set $S$,
\item which blue vertices already have neighbors in $D$, and
\item which blue vertices do not yet have neighbors in $D$.
\end{itemize}

Before presenting the actual information we maintain, we prove some properties.
\begin{property}\label{prop:dom-future}
Let $S$ be a minimal set such that $\gamma_{RB}(G\setminus S) \le \gamma$, and let $D$ be a minimum dominating set of $G\setminus S$. For any node $t$ of $T_G$, let $G_t$ be the graph represented by the clique-width expression of the subtree rooted at $t$. If there exists a vertex $v\in V_t\setminus S$ such that $v$ is not dominated by $D\cap V_t$, then we can assume that for all $w\in V_t$ with $\ell_{G_t} (v) = \ell_{G_t} (w)$, the following holds:
\begin{enumerate}
\item $w \notin S$, and
\item there exists a vertex $u \in N(w)\cap D$ (the same for all $w$, though not necessarily unique) and the edge $wu$ has not yet been introduced .
\end{enumerate}
\end{property}
\begin{proof}
Since $v \notin S$, it must eventually be dominated by some vertex $u \in D$.
Because this is not true in $G_t$, the edge $vu$ must be introduced at a node $b$ representing O.3, as this is the only way to add new edges to our graph. Moreover, $b$ must be an ancestor of $t$. We can conclude that at $b$, we add edges between every vertex labeled $\ell_{G_{b}}(v)=i$ and every vertex labeled $\ell_{G_{b}}(u)=j$.
    Second, since $\ell_{G_t} (v) = \ell_{G_t} (w)$, we have $\ell_{G_{t'}} (v) = \ell_{G_{t'}} (w)$ for any ancestor node $t'$ of $t$. This holds because no operation partially changes a label; thus, if two vertices share a label at any point, this remains true until the end of the construction. Therefore, at node $b$, we also introduce the edge $wu$. 
    Note that for any such $w$, we can use the same vertex $u$ to dominate them, as its definition does not depend on $w$. This completes the proof of the second statement. 

    For the first statement, suppose $w\in S$ and let $S'=S\setminus \{w\}$. Then $D$ is a red-blue dominating set of $G\setminus S'$ with $|D|\le \gamma$, so $\gamma_{RB}(G\setminus S')\le \gamma$. This contradicts the minimality of $S$, as $|S'|<|S|$.
\end{proof}

The previous property allows us to assume that for any graph $G_t$ associated with node $t$, the blue vertices of label $i$ in $G_t$ satisfy: 
\begin{itemize}
    \item either we decide which are dominated and which are included in $S$, 
    \item or all vertices labeled $i$ at this point will be dominated by one vertex $u$ (not necessarily in $V_t$), with edges between $u$ and these vertices being introduced at an ancestor node of $t$.
\end{itemize}

Before proceeding, we explain the intuition behind the algorithm. 
For any graph $G_t$ where $t \in T_G$, we solve the following problem. 

Given a labeled red-blue graph $G$ with $cw$ labels, two binary vectors $\alpha=\{a_1,\ldots,a_{cw}\}$, $\beta=\{b_1,\ldots,b_{cw}\}$, and an integer $m \in [n]$, determine the smallest integer $s$ such that: 
\begin{itemize}
    \item There exist two sets $S \subseteq B$ and $D \subseteq R\setminus S$ where:
    \begin{itemize}
        \item $|D|\le m$ and $|S| = s$, 
        \item if $a_i=1$, then for any $u \in B \cap \{ v \mid \ell_{G_t}(v)=i \}$, either $u \in S$ or $N[u] \cap D \neq \emptyset$, 
        \item if $b_i = 0$, then for all $u \in R \cap \{ v \mid \ell_{G_t}(v)=i \}$ we have that $u \notin D$,
        \item if $b_i = 1$, then there exists a $u \in R \cap \{ v \mid \ell_{G_t}(v)=i \}$ such that $u \in D$. 
    \end{itemize}    
\end{itemize}

Before explaining how we solve this, let us interpret the restrictions imposed by $\alpha$ and $\beta$. The $i^{th}$ position of $\alpha$, $a_i$, relates to blue vertices labeled $i$. Specifically, $a_i=0$ means blue vertices with label $i$ need not be dominated by $D$, while $a_i=1$ means they must either be dominated or deleted (i.e., included in $S$).
For $\beta$, $b_i=0$ indicates that no red vertex labeled $i$ is included in $D$, while $b_i=1$ indicates that at least one such vertex must be in $D$. The integer $m$ is the maximum allowed size of $D$, and the function returns the size of the smallest set $S$ that must be deleted for such a $D$ to exist.

We compute a function $f_t$ that takes vectors $\alpha$, $\beta$, and an integer $m$, returning the value of $s$ for the graph $G_t$ associated with node $t$ of $T_G$. 

Since $f_t$ provides answers for all possible inputs $\alpha$, $\beta$, $m$, we can use $f_r$ (where $r$ is the root of $T_G$) to determine whether $\mathcal{I}$ is a yes-instance.

We compute $f_t$ for all nodes $t$ of $T_G$ using dynamic programming. When considering a node $t$, we subscript the hypothetical sets $S$ and $D$ with $t$ to distinguish between nodes. If no pair $(D_t,S_t)$ satisfies the restrictions of $f_t(\alpha,\beta,m)$, the function returns $+\infty$.

The construction is inductive, starting with the vertex-introduce nodes (the leaves of $T_G$), which form the base case.

\paragraph{Vertex-Introduce Nodes.} Consider a vertex-introduce node $t$. Since these nodes are leaves, $V_t$ contains only the newly introduced vertex, say $v$, and the only label is $\ell_{G_t}(v)$. We consider several cases. 

\textbf{Case 1: $v\in R\setminus B$.} Here, $\gamma_{RB}(G_t)= 0$, so no vertices need to be included in $S$. The only reason not to return $0$ is if $\beta$ requires including vertices in $D_t$ that do not exist. Let $\beta_1 = \textbf{0}$ and $\beta_2 \in \{0,1\}^{cw}$ be the vector with all positions $0$ except the $\ell_{G_t}(v)$ position. We set:
\begin{itemize}
        \item $f_t(\alpha,\beta_1,m) = 0$, for all vectors $\alpha$ (as there is no blue vertex) and $m \in [n]\cup \{0\}$, 
        \item $f_t(\alpha,\beta_2,m)= 0$, for all vectors $\alpha$ (as there is no blue vertex) and $m \in [n]$, and 
        \item $f_t(\alpha,\beta,m)= +\infty$, for all the other inputs.
    \end{itemize}
    Indeed, $f_t(\alpha,\beta_2,0)= +\infty$ since the $\beta_2$ indicates that we have at least one vertex in $D_t$ while the third argument indicates that $|D_t|=0$ (which is invalid). Also, for any $\beta \notin \{\beta_1,\beta_2\}$ we need to include in $D_t$ vertices that do not exist.

\textbf{Case 2: $v\in B\setminus R$.} Here, $v \in B$. Since there are no red vertices, $f(\alpha,\beta,m)= +\infty$ for any $\alpha \in \{0,1\}^{cw}$, $\beta \neq \textbf{0}$, and $m \ge 0$. 
If $\beta = \textbf{0}$, $D_t$ must be empty. Let $i = \ell_{G_t}(v)$. Then:
\begin{itemize}
    \item $f(\alpha,\textbf{0},m)= 0$ for all $\alpha$ with $a_i=0$, 
    \item $f(\alpha,\textbf{0},m)= 1$ otherwise.
\end{itemize}
If $a_i = 0$, vertex $v$ need not be dominated, so we do not delete it. If $a_i = 1$, $v$ must be either dominated or deleted. Since $D_t$ is empty, $v$ must be included in $S_t$, so we return $1$.

\textbf{Case 3: $v \in B\cap R$.} Here, $\beta$ can only take two valid values:
\begin{itemize}
    \item $\beta=\beta_1$ where $b_i=1$ for $i=\ell_{G_t}(v)$ and $b_j=0$ otherwise (indicating $v \in D_t$),
    \item $\beta=\beta_2=\textbf{0}$ (indicating $v \notin D_t$).
\end{itemize}
For $\beta \notin\{\beta_1,\beta_2\}$, we set $f_t(\alpha,\beta,m) = +\infty$ because $\beta$ cannot be satisfied. 
For $\beta=\beta_1$, $v$ dominates itself, so $f_t(\alpha, \beta_1, m) = 0$ for $m\ge 1$, and $f_t(\alpha, \beta_1, m)= +\infty$ for $m=0$ (no valid $D_t$).
For $\beta=\beta_2$, we consider $\alpha$:
\begin{itemize}
    \item If $a_{\ell_{G_t}(v)}=1$, then $f_t(\alpha,\beta_2,m) = 1$ for all $m\ge 0$, as $v$ cannot be dominated in $G_t$.
    \item If $a_{\ell_{G_t}(v)}=0$, then $f_t(\alpha,\beta_2,m) = 0$ for all $m\ge 0$, as $v$ need not be dominated or deleted.
\end{itemize}
This completes the vertex-introduce nodes. 

\paragraph{Rename Nodes.} Consider a rename node $t$ with child node $t'$. Node $t$ takes $G_{t'}$ as input and renames label $i$ to $j$. We assume $f_{t'}$ has been computed and use it to compute $f_{t}$. Note that $G_t$ and $G_{t'}$ are the same graphs with different labels. 

To compute $f_t(\alpha,\beta,m)$, note that any $\beta= \{b_1,\ldots,b_{cw}\}$ with $b_i\neq 0$ is invalid because no vertices labeled $i$ remain after renaming. Thus, $f_t(\alpha,\beta,m) = +\infty$ for any such $\beta$. 

Now suppose $\beta = \{b_1,\ldots,b_{cw}\}$ with $b_i=0$.
For $(p,q)\in \{0,1\}^2$, define $\beta^{p,q}=\{b^{p,q}_1,\ldots,b^{p,q}_{cw}\}$ where $b^{p,q}_i= p$, $b^{p,q}_j= q$, and $b^{p,q}_k = b_k$ for $k \in [cw]\setminus \{i,j\}$. Similarly, define $\alpha^{p,q} = \{a^{p,q}_1,\ldots,a^{p,q}_{cw}\}$ where $a^{p,q}_i= p$, $a^{p,q}_j= q$, and $a^{p,q}_k = a_k$ for $k \in [cw]\setminus \{i,j\}$.
We consider two cases based on $a_j$.

\medskip 

\noindent\textbf{Case 1 ($a_j = 0$):} 
Here, blue vertices labeled $j$ in $G_t$ need not be dominated or deleted. Thus, we consider $\alpha^{0,0}$ for $G_{t'}$ (no need to dominate or delete vertices labeled $i$ or $j$). We further distinguish based on $\beta_j$.

\textbf{Case 1.a ($b_j = 0$):} Since $D_t$ cannot include vertices labeled $j$ in $G_t$, we consider $\beta^{0,0}$ for $G_{t'}$ (no vertices labeled $i$ or $j$ in $D_{t'}$). Thus, $f_t(\alpha,\beta,m) = f_{t'}(\alpha^{0,0},\beta^{0,0},m)$.

\textbf{Case 1.b ($b_j = 1$):} This occurs in one of three subcases:
\begin{enumerate}
    \item there exist $x,y\in D_t$ such that $\ell_{G_{t'}}(x)=i$ and $\ell_{G_{t'}}(y)=j$
        \item there exists $x\in D_t$ such that $\ell_{G_{t'}}(x)=i$ and $\ell_{G_{t'}}(y)\neq j$ for all $y\in D_t$, 
        \item there exists $y\in D_t$ such that $\ell_{G_{t'}}(y)=j$ and $\ell_{G_{t'}}(x)\neq i$ for all $x\in D_t$.
\end{enumerate}
Thus,
\begin{multline*}
f_{t}(\alpha,\beta,m) = \min \{f_{t'}(\alpha^{0,0},\beta^{1,1},m), \\f_{t'}(\alpha^{0,0},\beta^{1,0},m), f_{t'}(\alpha^{0,0},\beta^{0,1},m)\}
\end{multline*}

\medskip 

\noindent\textbf{Case 2 ($a_j = 1$):} 
Here, any blue vertex $v$ labeled $j$ in $G_t$ must either be dominated by $D_t$ or included in $S_t$. Since vertices labeled $i$ in $t'$ become labeled $j$ in $t$, we consider $\alpha^{1,1}$ for $G_{t'}$ (vertices labeled $i$ or $j$ must be dominated or deleted). We again distinguish based on $\beta_j$.

\textbf{Case 2.a ($\beta_j = 0$):} 
As in Case 1.a, we consider $\beta^{0,0}$ for $G_{t'}$. Thus, $f_t(\alpha,\beta,m)= f_{t'}(\alpha^{1,1}, \beta^{0,0},m)$.

\textbf{Case 2.b ($\beta_j = 1$):} As in Case 1.b, this occurs in one of three subcases:
\begin{enumerate}
    \item There exist $x,y\in D_t$ with $\ell_{G_{t'}}(x)=i$ and $\ell_{G_{t'}}(y)=j$,
    \item There exists $x\in D_t$ with $\ell_{G_{t'}}(x)=i$ and $\ell_{G_{t'}}(y)\neq j$ for all $y\in D_t$,
    \item There exists $y\in D_t$ with $\ell_{G_{t'}}(y)=j$ and $\ell_{G_{t'}}(x)\neq i$ for all $x\in D_t$.
\end{enumerate}
Therefore, 
\begin{multline*}
f_{t}(\alpha,\beta,m) = \min \{
f_{t'}(\alpha^{1,1},\beta^{1,1},m),\\
f_{t'}(\alpha^{1,1},\beta^{1,0},m),
f_{t'}(\alpha^{1,1},\beta^{0,1},m)
\}
\end{multline*}
This completes the rename nodes.

\paragraph{Edge-Introduce Nodes.} 

Consider an edge-introduce node $t$ with child node $t'$. Node $t$ takes $G_{t'}$ and two labels $i,j$ as input and adds all the edges between the vertices labeled $i$ and the vertices labeled $j$.

    We fix vectors $\alpha$, $\beta$ and a positive integer $m$. We will show how we compute $f_t(\alpha,\beta,m)$ by using $f_{t'}$. 
    Consider a label $\ell$. 
    Notice $\beta$ and $m$ impose the same restrictions on both $G_t$ and $G_{t'}$. The difference is that any vertex labeled $i$ (resp. labeled $j$) dominates all vertices labeled $j$ (resp. $i$) in $G_t$, while this is not necessarily true in $G_{t'}$.
    Therefore, the important difference is that if $D_t$ includes vertices labeled $i$ ($j$ resp.) and the vertices labeled $j$ (labeled $i$ resp.) need to be dominated or deleted (i.e., $a_i = 1$ and $a_j = 1$ resp.), we can consider two cases:
    \begin{itemize}
        \item either the vertices labeled $j$ ($i$ resp.) were dominated or deleted in $G_{t'}$,
        \item or the vertices labeled $j$ ($i$ resp.) were not (necessarily) dominated  in $G_{t'}$, while they are dominated in $G_t$ because of the newly introduced edges.
    \end{itemize}

\medskip 

    To sum up, having fixed $\beta$ and $m$, we have that $f_t(\alpha,\beta,m)=f_{t'}(\alpha',\beta,m)$ for a specific $\alpha'$ such that $a_\ell = a'_\ell$, for all $\ell \notin \{i,j\}$ (where $a'_\ell$ is the element in the $\ell^{th}$ position of $\alpha'$). Thus, it suffices to consider all the possibilities for this $\alpha'$. This is thanks to Property~\ref{prop:dom-future}.

\medskip 
\noindent\textbf{Case 1 ($a_i = a_j = 0$):} Here, vertices labeled $i$ or $j$ need not be dominated, so the new edges do not affect the solution. Thus, $f_t(\alpha,\beta,m) = f_{t'}(\alpha,\beta,m)$.

\medskip 

\noindent\textbf{Case 2 ($a_i=1$ and $a_j = 0$):} Either vertices labeled $i$ were already dominated or included in $S_t$ in $G_{t'}$, or $b_j = 1$ and they are dominated by vertices labeled $j$ due to the new edges. Let $\alpha'$ be identical to $\alpha$ except $a'_{i}=0$. Then:
\begin{itemize}
    \item if $b_j =0$, we set $f_t(\alpha,\beta,m) = f_{t'}(\alpha,\beta,m)$,
    \item otherwise, $f_t(\alpha,\beta,m)~=~\min \{ f_{t'}(\alpha,\beta,m), \\f_{t'}(\alpha',\beta,m) \}$.
\end{itemize}

\medskip 

\noindent\textbf{Case 3 ($a_j=1$ and $a_i = 0$):} Symmetric to Case 2. 

\medskip 

\noindent\textbf{Case 4 ($a_j=1$ and $a_i = 1$):} Here, we distinguish subcases based on $b_i$ and $b_j$. 

\textbf{Case 4.a ($b_i = b_j = 0$):}
Here, $D_t$ cannot include vertices labeled $i$ or $j$, so $f_t(\alpha,\beta,m) = f_{t'}(\alpha,\beta,m)$.

\textbf{Case 4.b ($b_i = 1$, $b_j = 0$):}
Vertices labeled $i$ cannot be dominated by vertices labeled $j$ (since $b_j=0$), but vertices labeled $j$ can be dominated by vertices labeled $i$. Let $\alpha'$ be identical to $\alpha$ except $a'_j=0$. Then, $f_t(\alpha,\beta,m) = \min \{ f_{t'}(\alpha,\beta,m), f_{t'}(\alpha',\beta,m) \}$.

\textbf{Case 4.c ($b_i = 0$, $b_j = 1$):} Symmetric to Case 4.b.

\textbf{Case 4.d ($b_i = 1$, $b_j = 1$):} Here, the minimal $s$ is achieved when:
\begin{itemize}
    \item[(I)] Vertices labeled $i$ are dominated by vertices labeled $j$ and vice versa,
    \item[(II)] Vertices labeled $i$ are dominated by vertices labeled $j$, but vertices labeled $j$ were already dominated in $G_{t'}$,
    \item[(III)] Vertices labeled $j$ are dominated by vertices labeled $i$, but vertices labeled $i$ were already dominated in $G_{t'}$,
    \item[(IV)] Both were already dominated in $G_{t'}$ (new edges do not help).
\end{itemize}
Define:
\begin{itemize}
    \item $\alpha^{(I)}$ identical to $\alpha$ except $a^{(I)}_i = a^{(I)}_j =0$,
    \item $\alpha^{(II)}$ identical to $\alpha$ except $a^{(II)}_i = 0$,
    \item $\alpha^{(III)}$ identical to $\alpha$ except $a^{(III)}_j = 0$.
\end{itemize}
Then:
\begin{multline*}
f_t(\alpha,\beta,m) = \min \{
f_{t'}(\alpha^{(I)},\beta,m),\\
f_{t'}(\alpha^{(II)},\beta,m),
f_{t'}(\alpha^{(III)},\beta,m),
f_{t'}(\alpha,\beta,m)
\}
\end{multline*}
This completes the edge-introduce nodes.

\paragraph{Union Nodes.} 
Consider a union node $t$ and let $t_1$, $t_2$ be the children of $t$ in $T_G$. 
    Recall that $t$ takes as input the labeled graphs $G_{t_1}$ and $G_{t_2}$ and constructs the disjoint union of the two.

    To ease the exposition, we will slightly abuse notations. In particular, we define $\lor : \{0,1\}^{cw} \times \{0,1\}^{cw} \rightarrow \{0,1\}^{cw}$ such that $\lor(\beta_1,\beta_2)=\beta$ where the $i^{th}$ position of $\beta$ is $0$ if and only if the $i^{th}$ position of $\beta_1$ and the $i^{th}$ position of $\beta_2$ are both equal to $0$. 
    
    Since $G_t$ is the disjoin union of $G_{t_1}$ and $G_{t_2}$, the function $f_{t}(\alpha,\beta,m)$ must return the value 
    $\min \{ f_{t_1}(\alpha,\beta_1,m_1) + f_{t_1}(\alpha,\beta_2,m_2) \mid m_1+m_2 = n , \lor(\beta_1, \beta_2) = \beta\}$. 

    This finishes the consideration of the union nodes, as well as the description of the computation of $f_t$ for every type of node $t$ of $T_G$.

\medskip
    
    To finish the proof, we need to prove that the algorithm runs in $4^{cw} |V(G)|^{O(1)}$. Since the we have $|V(G)|^{O(1)}$ nodes in the clique-width expression, we just need to prove that we can compute the function $f_t$, for all nodes $t$, in $4^{cw} |V(G)|^{O(1)}$ time. This is straightforward for all the nodes of the clique-width expression other than the non-union nodes. 
    
    Consider a union node $t$, and let $t_1$ and $t_2$ be the children nodes of $t$.
    We will show that we can use the fast subset convolution technique in order to compute all values of $f_{t}$ in time $4^{cw} |V(G)|^{O(1)}$. 
    For any $\beta = \{b_1,\ldots,b_{cw}\} \in \{0,1\}^{cw}$, we denote by $Y_\beta \subseteq [cw]$ the set of indices such that $b_i = 1$ if and only if $i \in Y_\beta$. 
    
    For each vector $\alpha\in \{0,1\}^{cw}$ and positive integers $m$, $m_1$, $m_2$ define two functions:
    \begin{itemize}
        \item $f^{\alpha,m_1,m_2}_t: \mathcal{P}(cw) \rightarrow \mathbb{N}$ and
        \item $f^{\alpha,m}_{t_j}: \mathcal{P}(cw) \rightarrow \mathbb{N}$, for $j \in \{1,2\}$,
    \end{itemize} 
    
    In particular, for any $Y_\beta \subseteq \mathcal{P}(cw)$ we are setting:    
    \begin{align*}
        & f^{\alpha,m}_{t_j}(Y_{\beta})= f_{t_j}(\alpha, \beta, m), \text{ for all } j \in \{1,2 \} , \text{and} \\
        & f^{\alpha,m_1,m_2}_t (Y_\beta) = \min \{ f_{t_1}(\alpha,\beta_1,m_1) + \\
        & \hspace{3cm} f_{t_2}(\alpha,\beta_2,m_2) \mid Y_{\beta_1} \cup Y_{\beta_2} = Y_\beta\}.
    \end{align*}

    Since we have already computed $f_{t_1}$ and $f_{t_2}$, we can assume that $f^{\alpha,m}_{t_j}(Y_{\beta})$ is given for all $\alpha$, $m$, $Y_{\beta}\subseteq [cw]$ and $j \in \{1,2\}$. We will show how we compute the values of $f^{\alpha,m_1,m_2}_t$ for fixed $\alpha$, $m_1$ and $m_2$. 

    Notice that
    \begin{align*}
    & f^{\alpha,m_1,m_2}_t (Y_\beta) \\
    & = \min \{ 
    f_{t_1}(\alpha,\beta_1,m_1) + f_{t_2}(\alpha,\beta_2,m_2) \\
    & \hspace{3cm} \mid (Y_{\beta_1} \cup Y_{\beta_2})= Y_\beta
    \} \\
    & = \min_{Y_{\beta_1} \cup Y_{\beta_2}= Y_\beta} 
    \{
    f^{\alpha,m_1}_{t_1}(Y_{\beta_1}) + f^{\alpha,m_2}_{t_2}(Y_{\beta_2}) 
    \}
    \end{align*}
    Assume that we have fixed $\alpha$, $m_1$, and $m_2$. Observe that the values of $f^{\alpha,m_1}_{t_1}$ and $f^{\alpha,m_2}_{t_2}$ are given for all possible inputs and are bounded by the size of the graph. Thus, we can compute all values of $f^{\alpha,m_1,m_2}_t$ in time $2^{cw}(|V(G_t)|)^{O(1)}$ by using the fast subset convolution~\cite{CyganFKLMPPS15}. 
    Since we have $2^{cw}(|V(G_t)|)^{O(1)}$ different triplets $\alpha,m_1,m_2$, we can compute all functions $\alpha,m_1,m_2$, $f^{\alpha,m_1,m_2}_t$, for all such triplets, in time $4^{cw}(|V(G_t)|)^{O(1)}$.

    Finally, we need to show how we compute the function $f_t(\alpha,\beta,m)$. Note that:
    \[
    f_t(\alpha,\beta,m) = \min_{m_1+m_2 = m} \{ f^{\alpha,m_1,m_2}_t (Y_\beta) \}
    \]

    At this point, it suffices to compute the value of the $f_t(\alpha,\beta,m)$ function for every possible triplet of $\alpha,\beta,m$, which is doable in $4^{cw}(|V(G_t)|^{O(1)}$.
\end{proof}
\fi

\ifshort
\begin{sketch}
We will perform dynamic programming over the clique-width expression $T_G$.
Let $G_t$ denote the graph constructed by the expression defined by the subtree rooted at node $t$, and let $V_t$ denote the vertex set of $G_t$. Our algorithm will compute a set $S\subseteq B$ of at most $k$ vertices such that there exists a set $D\subseteq R$ of at most $\gamma$ vertices that dominates all the blue vertices of $G\setminus S$, if such sets do indeed exist.  
For each $G_t$ associated with node $t$ of the given clique-width expression, the algorithm maintains information about:
\begin{itemize}
\item which red vertices are included in set $D$,
\item which blue vertices are included in set $S$,
\item which blue vertices already have neighbors in $D$, and
\item which blue vertices do not yet have neighbors in $D$.
\end{itemize}

Before presenting the actual information we maintain, we prove some properties.
\begin{property}\label{prop:dom-future}
Let $S$ be a minimal set such that $\gamma_{RB}(G\setminus S) \le \gamma$, and let $D$ be a minimum dominating set of $G\setminus S$. For any node $t$ of $T_G$, let $G_t$ be the graph represented by the clique-width expression of the subtree rooted at $t$. If there exists a vertex $v\in V_t\setminus S$ such that $v$ is not dominated by $D\cap V_t$, then we can assume that for all $w\in V_t$ with $\ell_{G_t} (v) = \ell_{G_t} (w)$, the following holds:
\begin{enumerate}
\item $w \notin S$, and
\item there exists a vertex $u \in N(w)\cap D$ (the same for all $w$, though not necessarily unique) and the edge $wu$ has not yet been introduced .
\end{enumerate}
\end{property}

Property~\ref{prop:dom-future} allows us to assume that for any graph $G_t$ associated with node $t$, the blue vertices of label $i$ in $G_t$ satisfy: 
\begin{itemize}
    \item either we decide which are dominated and which are included in $S$, 
    \item or all vertices labeled $i$ at this point will be dominated by one vertex $u$ (not necessarily in $V_t$), with edges between $u$ and these vertices being introduced at an ancestor node of $t$.
\end{itemize}

Before proceeding, we explain the intuition behind the algorithm. 
For any graph $G_t$ where $t \in T_G$, we solve the following problem. 

Given a labeled red-blue graph $G$ with $cw$ labels, two binary vectors $\alpha=\{a_1,\ldots,a_{cw}\}$, $\beta=\{b_1,\ldots,b_{cw}\}$, and an integer $m \in [n]$, determine the smallest integer $s$ such that: 
\begin{itemize}
    \item There exist two sets $S \subseteq B$ and $D \subseteq R\setminus S$ where:
    \begin{itemize}
        \item $|D|\le m$ and $|S| = s$, 
        \item if $a_i=1$, then for any $u \in B \cap \{ v \mid \ell_{G_t}(v)=i \}$, either $u \in S$ or $N[u] \cap D \neq \emptyset$, 
        \item if $b_i = 0$, then for all $u \in R \cap \{ v \mid \ell_{G_t}(v)=i \}$ we have that $u \notin D$,
        \item if $b_i = 1$, then there exists a $u \in R \cap \{ v \mid \ell_{G_t}(v)=i \}$ such that $u \in D$. 
    \end{itemize}    
\end{itemize}

Before explaining how we solve this, let us interpret the restrictions imposed by $\alpha$ and $\beta$. The $i^{th}$ position of $\alpha$, $a_i$, relates to blue vertices labeled $i$. Specifically, $a_i=0$ means blue vertices with label $i$ need not be dominated by $D$, while $a_i=1$ means they must either be dominated or deleted (i.e., included in $S$).
For $\beta$, $b_i=0$ indicates that no red vertex labeled $i$ is included in $D$, while $b_i=1$ indicates that at least one such vertex must be in $D$. The integer $m$ is the maximum allowed size of $D$, and the function returns the size of the smallest set $S$ that must be deleted for such a $D$ to exist.

We compute a function $f_t$ that takes vectors $\alpha$, $\beta$, and an integer $m$, returning the value of $s$ for the graph $G_t$ associated with node $t$ of $T_G$. 
Since $f_t$ provides answers for all possible inputs $\alpha$, $\beta$, $m$, we can use $f_r$ (where $r$ is the root of $T_G$) to determine whether $\mathcal{I}$ is a yes-instance.

The rest of the proof consists of the computation of $f_t$ for all nodes $t$ of $T_G$, which is done through dynamic programming. When considering a node $t$, we subscript the hypothetical sets $S$ and $D$ with $t$ to distinguish between nodes. If no pair $(D_t,S_t)$ satisfies the restrictions of $f_t(\alpha,\beta,m)$, the function returns $+\infty$.

The construction is inductive, starting with the vertex-introduce nodes (the leaves of $T_G$).

\paragraph{Vertex-Introduce Nodes.} Consider a vertex-introduce node $t$. Since these nodes are leaves, $V_t$ contains only the newly introduced vertex, say $v$, and the only label is $\ell_{G_t}(v)$. We consider several cases. 

\textbf{Case 1: $v\in R\setminus B$.} Here, $\gamma_{RB}(G_t)= 0$, so no vertices need to be included in $S$. The only reason not to return $0$ is if $\beta$ requires including vertices in $D_t$ that do not exist. Let $\beta_1 = \textbf{0}$ and $\beta_2 \in \{0,1\}^{cw}$ be the vector with all positions $0$ except the $\ell_{G_t}(v)$ position. We set:
\begin{itemize}
        \item $f_t(\alpha,\beta_1,m) = 0$, for all vectors $\alpha$ (as there is no blue vertex) and $m \in [n]\cup \{0\}$, 
        \item $f_t(\alpha,\beta_2,m)= 0$, for all vectors $\alpha$ (as there is no blue vertex) and $m \in [n]$, and 
        \item $f_t(\alpha,\beta,m)= +\infty$, for all the other inputs.
    \end{itemize}

\textbf{Case 2: $v\in B\setminus R$.} Here, $v \in B$. Since there are no red vertices, $f(\alpha,\beta,m)= +\infty$ for any $\alpha \in \{0,1\}^{cw}$, $\beta \neq \textbf{0}$, and $m \ge 0$. 
If $\beta = \textbf{0}$, $D_t$ must be empty. Let $i = \ell_{G_t}(v)$. Then:
\begin{itemize}
    \item $f(\alpha,\textbf{0},m)= 0$ for all $\alpha$ with $a_i=0$, 
    \item $f(\alpha,\textbf{0},m)= 1$ otherwise.
\end{itemize}

\textbf{Case 3: $v \in B\cap R$.} Here, $\beta$ can only take two valid values:
\begin{itemize}
    \item $\beta=\beta_1$ where $b_i=1$ for $i=\ell_{G_t}(v)$ and $b_j=0$ otherwise (indicating $v \in D_t$),
    \item $\beta=\beta_2=\textbf{0}$ (indicating $v \notin D_t$).
\end{itemize}
For $\beta \notin\{\beta_1,\beta_2\}$, we set $f_t(\alpha,\beta,m) = +\infty$. 
For $\beta=\beta_1$, we set $f_t(\alpha, \beta_1, m) = 0$ for $m\ge 1$, and $f_t(\alpha, \beta_1, m)= +\infty$ for $m=0$ (no valid $D_t$).
For $\beta=\beta_2$, we consider $\alpha$:
\begin{itemize}
    \item If $a_{\ell_{G_t}(v)}=1$, then $f_t(\alpha,\beta_2,m) = 1$ for all $m\ge 0$.
    \item If $a_{\ell_{G_t}(v)}=0$, then $f_t(\alpha,\beta_2,m) = 0$ for all $m\ge 0$.
\end{itemize}
This completes the vertex-introduce nodes. 

\paragraph{Rename Nodes.} Consider a rename node $t$ with child node $t'$. Node $t$ takes $G_{t'}$ as input and renames label $i$ to $j$. We assume $f_{t'}$ has been computed and use it to compute $f_{t}$. Note that $G_t$ and $G_{t'}$ are the same graphs with different labels. 

To compute $f_t(\alpha,\beta,m)$, note that any $\beta= \{b_1,\ldots,b_{cw}\}$ with $b_i\neq 0$ is invalid because no vertices labeled $i$ remain after renaming. Thus, $f_t(\alpha,\beta,m) = +\infty$ for any such $\beta$. 

Now suppose $\beta = \{b_1,\ldots,b_{cw}\}$ with $b_i=0$.
For $(p,q)\in \{0,1\}^2$, define $\beta^{p,q}=\{b^{p,q}_1,\ldots,b^{p,q}_{cw}\}$ where $b^{p,q}_i= p$, $b^{p,q}_j= q$, and $b^{p,q}_k = b_k$ for $k \in [cw]\setminus \{i,j\}$. Similarly, define $\alpha^{p,q} = \{a^{p,q}_1,\ldots,a^{p,q}_{cw}\}$ where $a^{p,q}_i= p$, $a^{p,q}_j= q$, and $a^{p,q}_k = a_k$ for $k \in [cw]\setminus \{i,j\}$.
We consider two cases based on $a_j$.

\medskip 

\noindent\textbf{Case 1 ($a_j = 0$):} 
Here, blue vertices labeled $j$ in $G_t$ need not be dominated or deleted. Thus, we consider $\alpha^{0,0}$ for $G_{t'}$. We further distinguish based on $\beta_j$.

\textbf{Case 1.a ($b_j = 0$):} Since $D_t$ cannot include vertices labeled $j$ in $G_t$, we consider $\beta^{0,0}$ for $G_{t'}$ (no vertices labeled $i$ or $j$ in $D_{t'}$). Thus, $f_t(\alpha,\beta,m) = f_{t'}(\alpha^{0,0},\beta^{0,0},m)$.

\textbf{Case 1.b ($b_j = 1$):} This occurs in one of three subcases:
\begin{enumerate}
    \item there exist $x,y\in D_t$ such that $\ell_{G_{t'}}(x)=i$ and $\ell_{G_{t'}}(y)=j$
        \item there exists $x\in D_t$ such that $\ell_{G_{t'}}(x)=i$ and $\ell_{G_{t'}}(y)\neq j$ for all $y\in D_t$, 
        \item there exists $y\in D_t$ such that $\ell_{G_{t'}}(y)=j$ and $\ell_{G_{t'}}(x)\neq i$ for all $x\in D_t$.
\end{enumerate}
Thus,
\begin{multline*}
f_{t}(\alpha,\beta,m) = \min \{f_{t'}(\alpha^{0,0},\beta^{1,1},m),\\ 
f_{t'}(\alpha^{0,0},\beta^{1,0},m), f_{t'}(\alpha^{0,0},\beta^{0,1},m)\}  
\end{multline*}

\smallskip 

\noindent\textbf{Case 2 ($a_j = 1$):} 
Here, any blue vertex $v$ labeled $j$ in $G_t$ must either be dominated by $D_t$ or included in $S_t$. Since vertices labeled $i$ in $t'$ become labeled $j$ in $t$, we consider $\alpha^{1,1}$ for $G_{t'}$. We again distinguish based on $\beta_j$.

\textbf{Case 2.a ($\beta_j = 0$):} 
As in Case 1.a, we consider $\beta^{0,0}$ for $G_{t'}$. Thus, $f_t(\alpha,\beta,m)= f_{t'}(\alpha^{1,1}, \beta^{0,0},m)$.

\textbf{Case 2.b ($\beta_j = 1$):} Similarly to Case 1.b.
\begin{multline*}
f_{t}(\alpha,\beta,m) = \min \{f_{t'}(\alpha^{1,1},\beta^{1,1},m),\\ 
f_{t'}(\alpha^{1,1},\beta^{1,0},m), f_{t'}(\alpha^{1,1},\beta^{0,1},m)\}
\end{multline*}
This completes the rename nodes.

\paragraph{Edge-Introduce Nodes.} 

Consider an edge-introduce node $t$ with child node $t'$. Node $t$ takes $G_{t'}$ and two labels $i,j$ as input and adds all the edges between the vertices labeled $i$ and the vertices labeled $j$.

    We fix vectors $\alpha$, $\beta$ and a positive integer $m$. We will show how we compute $f_t(\alpha,\beta,m)$ by using $f_{t'}$. 
    Consider a label $\ell$. 
    Notice $\beta$ and $m$ impose the same restrictions on both $G_t$ and $G_{t'}$. The difference is that any vertex labeled $i$ (resp. labeled $j$) dominates all vertices labeled $j$ (resp. $i$) in $G_t$, while this is not necessarily true in $G_{t'}$.
    Therefore, the important difference is that if $D_t$ includes vertices labeled $i$ ($j$ resp.) and the vertices labeled $j$ (labeled $i$ resp.) need to be dominated or deleted (i.e., $a_i = 1$ and $a_j = 1$ resp.), we can consider two cases:
    \begin{itemize}
        \item either the vertices labeled $j$ ($i$ resp.) were dominated or deleted in $G_{t'}$,
        \item or the vertices labeled $j$ ($i$ resp.) were not (necessarily) dominated  in $G_{t'}$, while they are dominated in $G_t$ because of the newly introduced edges.
    \end{itemize}

\medskip 

    To sum up, having fixed $\beta$ and $m$, we have that $f_t(\alpha,\beta,m)=f_{t'}(\alpha',\beta,m)$ for a specific $\alpha'$ such that $a_\ell = a'_\ell$, for all $\ell \notin \{i,j\}$ (where $a'_\ell$ is the element in the $\ell^{th}$ position of $\alpha'$). Thus, it suffices to consider all the possibilities for this $\alpha'$. This is thanks to Property~\ref{prop:dom-future}.

\paragraph{Union Nodes.} 
Consider a union node $t$ and let $t_1$, $t_2$ be the children of $t$ in $T_G$. 
    Recall that $t$ takes as input the labeled graphs $G_{t_1}$ and $G_{t_2}$ and constructs their disjoint union.

    To ease the exposition, we will slightly abuse notations. In particular, we define $\lor : \{0,1\}^{cw} \times \{0,1\}^{cw} \rightarrow \{0,1\}^{cw}$ such that $\lor(\beta_1,\beta_2)=\beta$ where the $i^{th}$ position of $\beta$ is $0$ if and only if the $i^{th}$ position of $\beta_1$ and the $i^{th}$ position of $\beta_2$ are both equal to $0$. 
    
    Since $G_t$ is the disjoin union of $G_{t_1}$ and $G_{t_2}$, the function $f_{t}(\alpha,\beta,m)$ must return the value 
    $\min \{ f_{t_1}(\alpha,\beta_1,m_1) + f_{t_1}(\alpha,\beta_2,m_2) \mid m_1+m_2 = n , \lor(\beta_1, \beta_2) = \beta\}$. 

    This finishes the consideration of the union nodes, as well as the description of the computation of $f_t$ for every type of node $t$ of $T_G$.

\medskip
    
    To finish the proof, we need to prove that the algorithm runs in $4^{cw} |V(G)|^{O(1)}$. Since the we have $|V(G)|^{O(1)}$ nodes in the clique-width expression, we just need to prove that we can compute the function $f_t$, for all nodes $t$, in $4^{cw} |V(G)|^{O(1)}$ time. This is straightforward for all the nodes of the clique-width expression other than the union nodes. 
    Achieving this running time for the union nodes requires some efficient computations, which are achieved thanks to the fast subset convolution technique~\cite{CyganFKLMPPS15}. 
\end{sketch}
\fi  

Finally, unless the \textsc{Strong Exponential Time Hypothesis}(\textsc{SETH})~\cite{IP01} fails, the \textsc{Dominating Set} problem cannot be solved in $(4 -\varepsilon)^{cw} n^{O(1)}$ time~\cite{KatsikarelisLP19}. In view of the connection between this problem and \mrbcShort{}, it follows that the algorithm presented in Theorem~\ref{thm:cw} is asymptotically optimal (under the \textsc{SETH}).

\section{Conclusion}
Despite the plethora of optimization problems that have been introduced to study the relationships between service providers and clients, we found no prior work that formally models the scenario where some clients also serve as providers for other clients. To overcome this limitation, we introduced the \mrbcShort{} problem and initiated its formal study. Our work opens two main promising research directions. The first concerns the efficient algorithms we present here. Indeed, it would be interesting to implement our algorithms and check their performance in practice. The second path to follow stems from the observation that solving the \mrbcShort{} problem is at least as hard as solving the \textsc{Dominating Set}; thus, there is ample motivation to develop heuristic solutions for \mrbcShort{}.

\bibliography{aaai2026}

\end{document}